\documentclass[11pt]{article}

\usepackage{amssymb,amsmath,amsfonts}
\usepackage{amsthm}
\usepackage{subfig}
\usepackage{times,inconsolata}
\usepackage{fullpage}
\usepackage{latexsym}
\usepackage{verbatim}
\usepackage{enumerate}
\usepackage{graphicx}
\usepackage{url}
\usepackage[ruled,linesnumbered,noend]{algorithm2e}
\usepackage{color}

\newcommand{\Ch}{\Children}
\newcommand{\Children}{\ensuremath{{\mathrm{Ch}}}}

\newcommand{\Cl}{\ensuremath{{\mathrm{Cl}}}} 		

\newcommand{\MP}{\ensuremath{M_\P}} 
\newcommand{\incompatible}{\texttt{incompatible}}
\newcommand{\Uinit}{\ensuremath{U_\mathrm{root}}} 
\newcommand{\Yinit}{\ensuremath{Y_\mathrm{root}}} 
\newcommand{\Desc}{\ensuremath{\mathrm{Desc}}} 

\newcommand{\DG}{\ensuremath{H_\P}} 
\newcommand{\GBNT}{\ensuremath{G_{\mathtt{BNT}}}} 

\newcommand{\SLBuild}{\ensuremath{\texttt{BuildNT}}} 

\newcommand{\InitLists}{\ensuremath{\texttt{Initialize}}}
\newcommand{\DelLabel}{\ensuremath{\texttt{Delete}}}

\renewcommand{\P}{\ensuremath{\mathcal{P}}}

\newcommand{\T}{\ensuremath{\mathcal{T}}}

\newcommand{\cnt}{\ensuremath{\mathtt{count}}} 
\newcommand{\SEMI}{\ensuremath{\mathtt{semiU}}} 
\newcommand{\WEIGHT}{\ensuremath{\mathtt{weight}}} 
\newcommand{\NULL}{\ensuremath{\mathtt{null}}} 
\newcommand{\MAP}{\ensuremath{\mathtt{map}}} 

\newtheorem{theorem}{Theorem}
\newtheorem{lemma}{Lemma}

\newtheorem{observation}{Observation}

\theoremstyle{definition}

\SetKw{continue}{continue}

\begin{document}
\title{Fast Compatibility Testing for Phylogenies with Nested Taxa\thanks{Supported in part by the National Science Foundation under grant CCF-1422134.}}
\author{
Yun Deng\thanks{Department of Computer Science, Iowa State University, Ames, IA 50011, USA, {\tt yundeng@iastate.edu}}
\and
David Fern\'{a}ndez-Baca\thanks{Department of Computer Science, Iowa State University, Ames, IA 50011, USA, {\tt fernande@iastate.edu}. 
}
}

\date{\empty}

\maketitle

\begin{abstract}
Semi-labeled trees are phylogenies whose internal nodes may be labeled by higher-order taxa.  Thus, a leaf labeled \emph{Mus musculus} could nest within a subtree whose root node is labeled Rodentia, which itself could nest within a subtree whose root is labeled Mammalia.  
Suppose we are given collection $\P$ of semi-labeled trees over various subsets of a set of taxa.  The ancestral compatibility problem asks whether there is a semi-labeled tree $\T$ that respects the clusterings and the ancestor/descendant relationships implied by the trees in $\P$. 
%
We give a $\tilde{O}(\MP)$ algorithm for the ancestral compatibility problem, where $\MP$ is the total number of nodes and edges in the trees in $\P$.  Unlike the best previous algorithm, the running time of our method does not depend on the degrees of the nodes in the input trees.  
\end{abstract}

\section{Introduction}

In the \emph{tree compatibility problem}, we are given a collection $\P = \{\T_1, \ldots, \T_k\}$ of rooted phylogenetic trees with partially overlapping taxon sets. $\P$ is called a \emph{profile} and the trees in $\P$ are the \emph{input trees}.  The question is whether there exists a tree $\T$ whose taxon set is the union of the taxon sets of the input trees, such that $\T$ exhibits the clusterings implied by the input trees.  That is, if two taxa are together in a subtree of some input tree, then they must also be together in some subtree of $\T$. 
The tree compatibility problem has been studied for over three decades \cite{Aho81a,DengFB2016,HenzingerKingWarnow99,Steel92}.

In the original version of the tree compatibility problem, only the leaves of the input trees are labeled. 
Here we study a generalization, called \emph{ancestral compatibility}, in which taxa may be \emph{nested}.  That is, the internal nodes may also be labeled; these labels represent \emph{higher-order taxa}, which are, in effect, sets of taxa.  Thus, for example,   
an input tree may contain the taxon \emph{Glycine max} (soybean) nested within a subtree whose root is labeled Fabaceae (the legumes), itself nested within an Angiosperm subtree.  Note that leaves themselves may be labeled by higher-order taxa.  The question now is whether there is a tree $\T$  whose taxon set is the union of the taxon sets of the input trees, such that $\T$ exhibits not only the clusterings among the taxa, but also the ancestor/descendant relationships among taxa in the input trees.  Our main result is a $\tilde{O}(\MP)$ algorithm for the compatibility problem for trees with nested taxa, where $\MP$ is the total number of nodes and edges in the trees in $\P$.

\vspace{-1.5\parsep}

\paragraph{Background.}  The tree compatibility problem is a basic special case of the \emph{supertree problem}.
A supertree method is a way to synthesize a collection of phylogenetic trees with partially overlapping taxon sets into a single supertree that represents the information in the input trees.  The supertree approach, proposed
in the early 90s \cite{Baum:1992,Ragan:1992}, has been used successfully to build large-scale phylogenies \cite{BinindaEmonds:Nature:07}. 

The original supertree methods were limited to input trees where only the leaves are labeled.
Page \cite{Page2004} was among the first to note the need to handle phylogenies where internal nodes are labeled, and taxa are nested.  A major motivation is the desire to incorporate \emph{taxonomies} as input trees in large-scale supertree analyses, as way to circumvent one of the obstacles to building comprehensive phylogenies: the limited taxonomic overlap among different phylogenetic studies \cite{Sanderson:2008}. Taxonomies group organisms according to a system of taxonomic rank  (e.g., family, genus, and species); two examples are the NCBI taxonomy \cite{NCBI2009} and the Angiosperm taxonomy \cite{APG2016}.  Taxonomies spanning a broad range of taxa provide structure and completeness that might be hard to obtain otherwise.    
A recent example of the utility of taxonomies is the Open Tree of Life, a draft phylogeny for over 2.3 million species \cite{HinchliffPNAS2015}.

Taxonomies are not, strictly speaking, phylogenies.  In particular, their internal nodes and some of their leaves are labeled with higher-order taxa. Nevertheless, taxonomies have many of the same mathematical characteristics as phylogenies.  Indeed, both phylogenies and taxonomies are \emph{semi-labeled trees} 
\cite{BordewichEvansSemple2006,SempleSteel03}.  We will use this term throughout the rest of the paper to refer to trees with nested taxa.


The fastest previous algorithm for testing ancestral compatibility, based on earlier work by Daniel and Semple \cite{DanielSemple2004}, is due to Berry and Semple \cite{BerrySemple2006}. Their algorithm runs in $O \left (\log^2 n \cdot \tau_\P \right)$ time using $O \left ( \tau_\P \right)$ space. Here, $n$ is the number of distinct taxa in $\P$ and  $\tau_\P = \sum_{i = 1}^k \sum_{v \in I(\T_i)} d(v)^2$, where $I(\T_i)$ is the set of internal nodes of $\T_i$, for each $i \in \{1, \dots , k\}$, and $d(v)$ is the degree of node $v$.   While the algorithm is polynomial, its dependence on node degrees is problematic:  semi-labeled trees can be highly unresolved (i.e., contain nodes of high degree), especially if they are taxonomies.

\vspace{-1.5\parsep}

\paragraph{Our contributions.}
The $\tilde{O}(\MP)$ running time of our ancestral compatibility algorithm is independent of the degrees of the nodes of the input trees, a valuable characteristic for large datasets that include taxonomies.  To achieve this time bound, we extend ideas from our recent algorithm for testing the compatibility of ordinary phylogenetic trees \cite{DengFB2016}.   As in that algorithm, a central notion in the current paper is the \emph{display graph} of profile $\P$, denoted $\DG$.  This is the graph obtained from the disjoint union of the trees in $\P$ by identifying nodes that have the same label (see Section \ref{sec:testAC}).
The term ``display graph'' was introduced by Bryant and Lagergren  \cite{BryantLagergren06}, but similar ideas have been used elsewhere.  In particular, the display graph is closely related to Berry and Semple's \emph{restricted descendancy graph} \cite{BerrySemple2006}, a mixed graph whose directed edges correspond to the (undirected) edges of $\DG$ and whose undirected edges have no correspondence in $\DG$.  The second kind of edges are the major component of the $\tau_\P$ term in the time and space complexity of Berry and Semple's algorithm.  The absence of such edges makes $\DG$ significantly smaller than the restricted descendancy graph.  Display graphs also bear some relation to  \emph{tree alignment graphs} 
\cite{Smith:PloSCB:2013}.

Here, we exploit the display graph more extensively and more directly than our previous work.  Although the display graph of a collection of semi-labeled trees is more complex than that of a collection of  ordinary phylogenies, we are able to extend several of the key ideas --- notably, that of a semi-universal label --- to the general setting of semi-labeled trees.  As in \cite{DengFB2016}, the implementation relies on a dynamic graph data structure, but it requires a more careful amortized analysis based on a weighing scheme.

\vspace{-1.5\parsep}

\paragraph{Contents.} Section \ref{sec:prelims} presents basic definitions regarding semi-labeled trees and ancestral compatibility.  .  Section \ref{sec:dispGraph} introduces the display graph and discusses its properties.  Section \ref{sec:testAC} presents \SLBuild, our algorithm for testing ancestral compatibility.  Section \ref{sec:implementation} gives the implementation details for \SLBuild.  Section \ref{sec:discussion} gives some concluding remarks.

\section{Preliminaries}\label{sec:prelims}

For each positive integer $r$, $[r]$ denotes the set $\{1, \dots , r\}$.

Let $G$ be a graph. $V(G)$ and $E(G)$ denote the node and edge sets of $G$. The \emph{degree} of a node $v \in V(G)$ is the number of edges incident on $v$. A \emph{tree} is an acyclic connected graph.
In this paper, all trees are assumed to be rooted.  For a tree $T$, $r(T)$ denotes the root of $T$. 
Suppose $u, v \in V(T)$.  Then, $u$ is an \emph{ancestor} of $v$ in $T$, denoted $u \le_T v$, if $u$ lies on the path from $v$ to $r(T)$ in $T$.  If $u \le_T v$, then $v$ is a \emph{descendant} of $u$.  Node %
$u$ is a \emph{proper descendant} of $v$ if $u$ is a descendant of $v$ and $v \neq u$. If $\{u,v\} \in E(T)$ and $u \le_T v$, then  $u$ is the \emph{parent} of $v$ and $v$ is a \emph{child} of $u$. If neither $u \le_T v$ nor  $v \le_T u$ hold, then  we write $u \parallel_T v$ and say that $u$ and $v$ are \emph{not comparable} in $T$. 

\vspace{-1.5\parsep}

\paragraph{Semi-labeled trees.} A \emph{semi-labeled tree} is a pair $\T = (T,\phi)$ where $T$ is a tree and $\phi$ is a mapping from a set $L(\T)$ to $V(T)$ such that, for every node $v \in V(T)$ of degree at most two, $v \in \phi(L(\T))$.  $L(\T)$ is the \emph{label set} of $\T$ and $\phi$ is the \emph{labeling function} of $\T$.  

For every node $v \in V(T)$, $\phi^{-1}(v)$ denotes the (possibly empty) subset of $L(\T)$ whose elements map into $v$; these elements as the \emph{labels of $v$} (thus, each label is a taxon). If $\phi^{-1}(v) \neq \emptyset$, then $v$ is \emph{labeled}; otherwise, $v$ is \emph{unlabeled}.  
Note that, by definition, every leaf in a semi-labeled tree is labeled.  Further, any node, including the root, that has a single child must be labeled.  Nodes with two or more children may be labeled or unlabeled.  A semi-labeled tree $\T = (T,\phi)$ is \emph{singularly labeled} if every node in $T$ has
at most one label; $\T$ is \emph{fully labeled} if every node in $T$ is labeled.

Semi-labeled trees, also known as \emph{$X$-trees}, generalize ordinary phylogenetic trees, also known as \emph{phylogenetic $X$-trees} \cite{SempleSteel03}.  An ordinary phylogenetic tree is a semi-labeled tree $\T = (T,\phi)$ where $r(T)$ has degree at least two and $\phi$ is a bijection from $L(\T)$ into leaf set of $T$ (thus, internal nodes are not labeled). 

Let $\T = (T,\phi)$ be a semi-labeled tree and let $\ell$ and $\ell'$ be two labels in $L(\T)$.  
If $\phi(\ell) \le_T \phi(\ell')$, then we write $\ell \le_\T \ell'$, and say that $\ell'$ is a \emph{descendant} of $\ell$ in $\T$ and that $\ell$ is an \emph{ancestor} of $\ell'$.  We write $\ell <_\T \ell'$ if $\phi(\ell')$ is a proper descendant of $\phi(\ell)$.   If $\phi(\ell) \parallel_T \phi(\ell')$, then we write $\ell \parallel_\T \ell'$ and say that $\ell$ and $\ell'$ are \emph{not comparable} in $\T$.  If $\T$ is fully labeled and
$\phi(\ell)$ is the parent of $\phi(\ell')$ in $T$, then $\ell$ is the \emph{parent} of $\ell'$ in $\T$ and $\ell'$ is a \emph{child} of $\ell$ in $\T$; two labels with the same parent are \emph{siblings}.

Two semi-labelled trees $\T = (T,\phi)$ and $\T' = (T', \phi')$ are \emph{isomorphic} if there exists a bijection $\psi : V(T) \rightarrow V(T')$ such that $\phi' = \psi \circ \phi$ and, for any two nodes $u, v \in V(T)$, $(u,v) \in E(T)$ if and only $(\psi(u), \psi(v)) \in E(T')$.


Let  $\T = (T,\phi)$ be a semi-labeled tree. For each $u \in V(T)$,  $X(u)$ denotes the set of all labels in the subtree of $T$ rooted at $u$; that is, $X(u) = \bigcup_{v: u \le_T v} \phi^{-1}(v)$. $X(u)$ is called a \emph{cluster} of $T$.  
$\Cl(\T)$ denotes the set of all clusters of $\T$.   
It is well known \cite[Theorem 3.5.2]{SempleSteel03} that a semi-labeled tree $\T$ is completely determined by $\Cl(\T )$.  That is, if $\Cl(\T) = \Cl(\T')$ for some other semi-labeled tree $\T'$, then $\T$ is isomorphic to $\T'$.


Suppose $A \subseteq L(\T)$ for a semi-labeled tree $\T = (T,\phi)$.  The \emph{restriction} of $\T$ to $A$, denoted $\T |A$, is the semi-labeled tree whose cluster set is
$\Cl(\T | A) = \{X \cap A : X \in \Cl(\T) \text{ and } X \cap A \neq \emptyset \}.$
Intuitively, $\T | A$ is obtained from the minimal rooted subtree of $T$ that connects the nodes in $\phi(A)$ by suppressing all vertices of degree two that are not in $\phi(A)$.

Let $\T = (T,\phi)$ and $\T' = (T', \phi')$ be semi-labeled trees such that $L(\T') \subseteq L(\T)$.  $\T$  \emph{ancestrally displays} $\T'$ if $\Cl(\T') \subseteq \Cl(\T|L(\T'))$.  Equivalently, $\T$ ancestrally displays $\T'$ if $\T'$  can be obtained from $\T | L(\T')$ by contracting edges, and, for any $\ell_1, \ell_2 \in L(\T')$, 
(i) if $\ell_1 <_\mathcal{\T'} \ell_2$, then $\ell_1 <_\T \ell_2$, and
(ii)
if $\ell_1 \parallel_\mathcal{\T'} \ell_2$, then $\ell_1 \parallel_\mathcal{\T} \ell_2$.
The notion of ``ancestrally displays'' for semi-labeled trees generalizes the well-known notion of ``displays'' for ordinary phylogenetic trees \cite{SempleSteel03}.  


For a semi-labelled tree $\T$, let
\begin{equation}
D(\T) = \{(\ell,\ell'): \ell, \ell' \in L(\T) \text{ and } \ell <_\T \ell'\} \quad \text{ and } \quad
N(\T ) = \{\{\ell,\ell'\} : \ell, \ell' \in L(\T) \text{ and } \ell \parallel_\T \ell'\}.
\end{equation}
Note that $D(\T)$ consists of \emph{ordered} pairs, while $N(\T)$ consists of \emph{unordered} pairs.

\begin{lemma}[Bordewich et al.\  \cite{BordewichEvansSemple2006}]\label{lm:DN}
Let $\T$ and $\T'$ be semi-labelled trees such that $L(\T') \subseteq L(\T)$. Then $\T$ ancestrally displays $\T'$ if and only if $D(\T') \subseteq D(\T)$ and $N(\T') \subseteq N(\T)$.
\end{lemma}

\vspace{-1.5\parsep}

\paragraph{Profiles and ancestral compatibility.}
Throughout the rest of this paper $\P = \{\T_1, \dots, \T_k\}$ denotes a set where, for each $i \in [k]$, $\T_i = (T_i, \phi_i)$ is a semi-labeled tree. We refer to $\P$ as a \emph{profile}, and write $L(\P)$ to denote $\bigcup_{i\in[k]} L(\T_i)$, the \emph{label set} of $\P$.   Figure~\ref{fig:profile} shows a profile where $L(\P) = \{a,b,c,d,e,f,g,h,i\}$.
We write $V(\P)$ for $\bigcup_{i\in[k]} V(T_i)$ and $E(\P)$ for $\bigcup_{i\in[k]} E(T_i)$, 
The \emph{size} of $\P$ is $\MP = |V(\P)| + |E(\P)|$.  


\begin{figure}
\centering
\begin{minipage}[b]{0.46\linewidth}
  \includegraphics[scale=0.25]{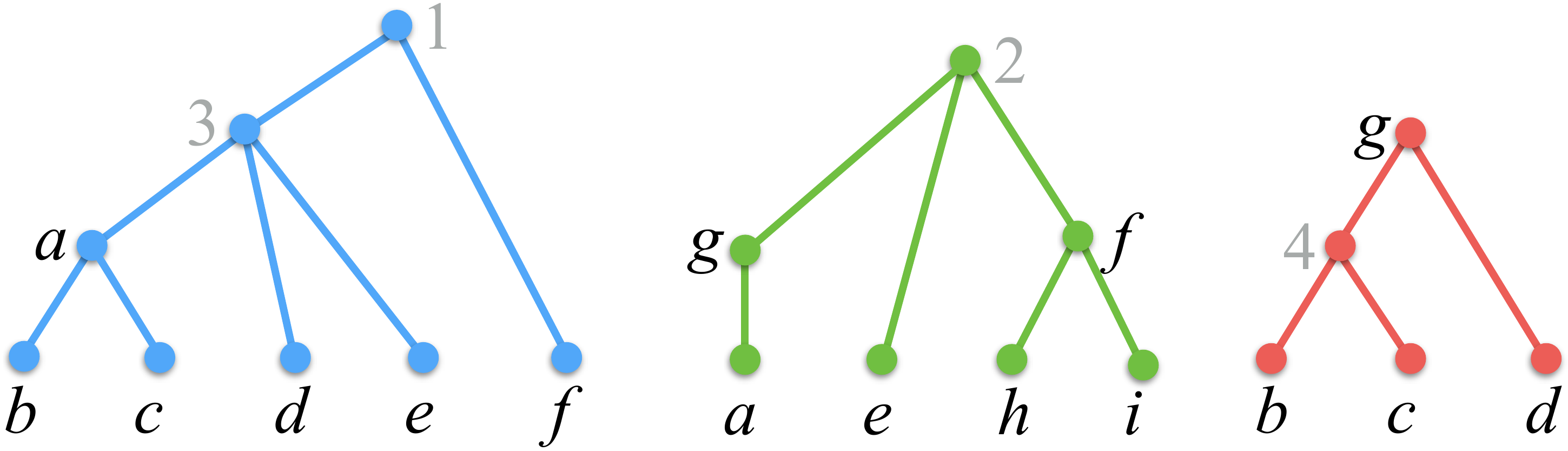}
  \caption{A profile $\P= \{\T_1, \T_2, \T_3\}$ --- trees are ordered left-to-right. 
The letters are the original labels; grey numbers are labels added to make the trees fully labeled. (Adapted from  \cite{BerrySemple2006}.)}
  \label{fig:profile}
\end{minipage}
\quad
\begin{minipage}[b]{0.24\linewidth}
  \includegraphics[scale=0.25]{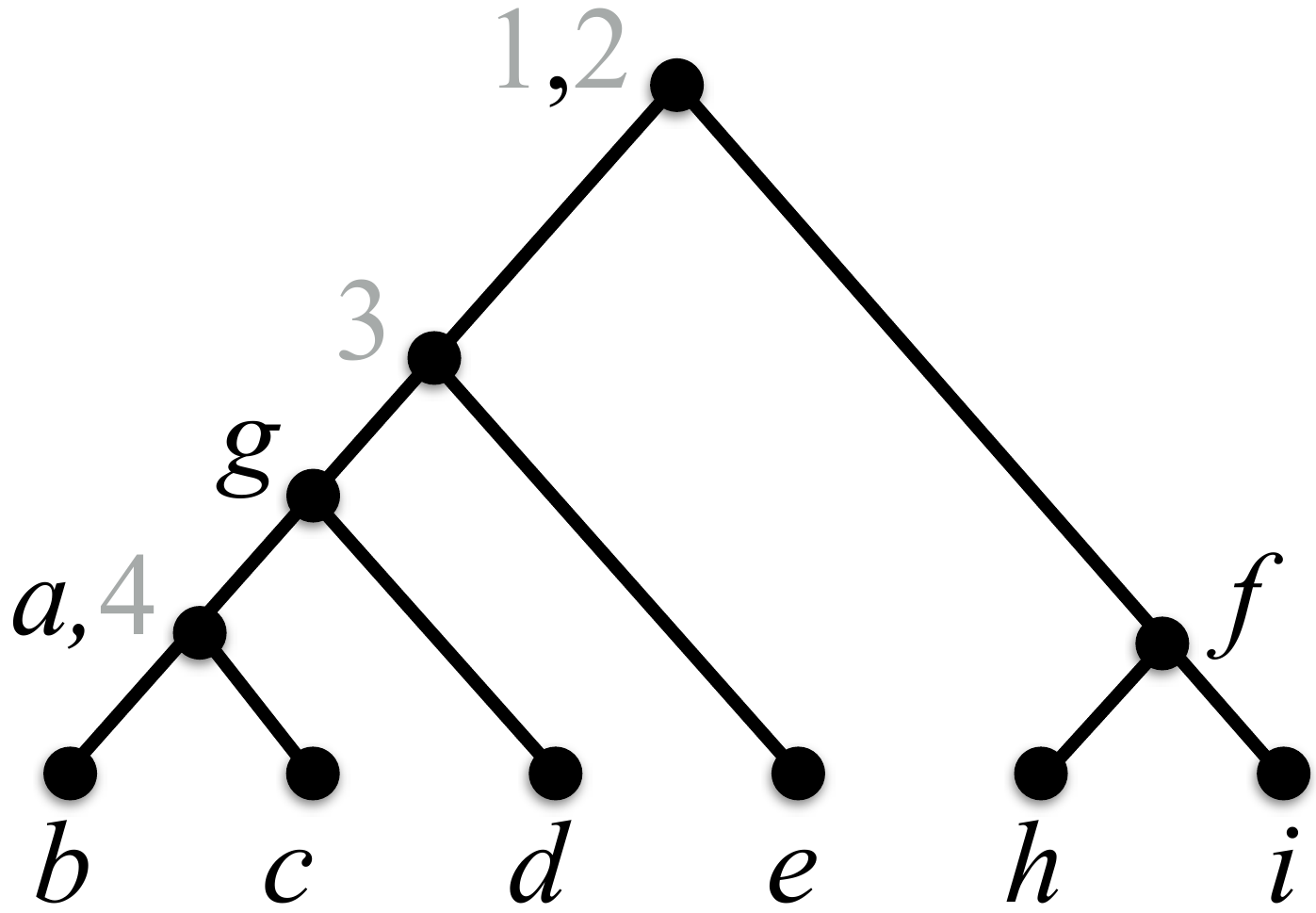}
  \caption{A tree $\T$ that ancestrally displays the profile of Figure \ref{fig:profile}. (Adapted from  \cite{BerrySemple2006}.)}
  \label{fig:supertree}
\end{minipage}
\quad
\begin{minipage}[b]{0.23\linewidth}
  \includegraphics[scale=0.25]{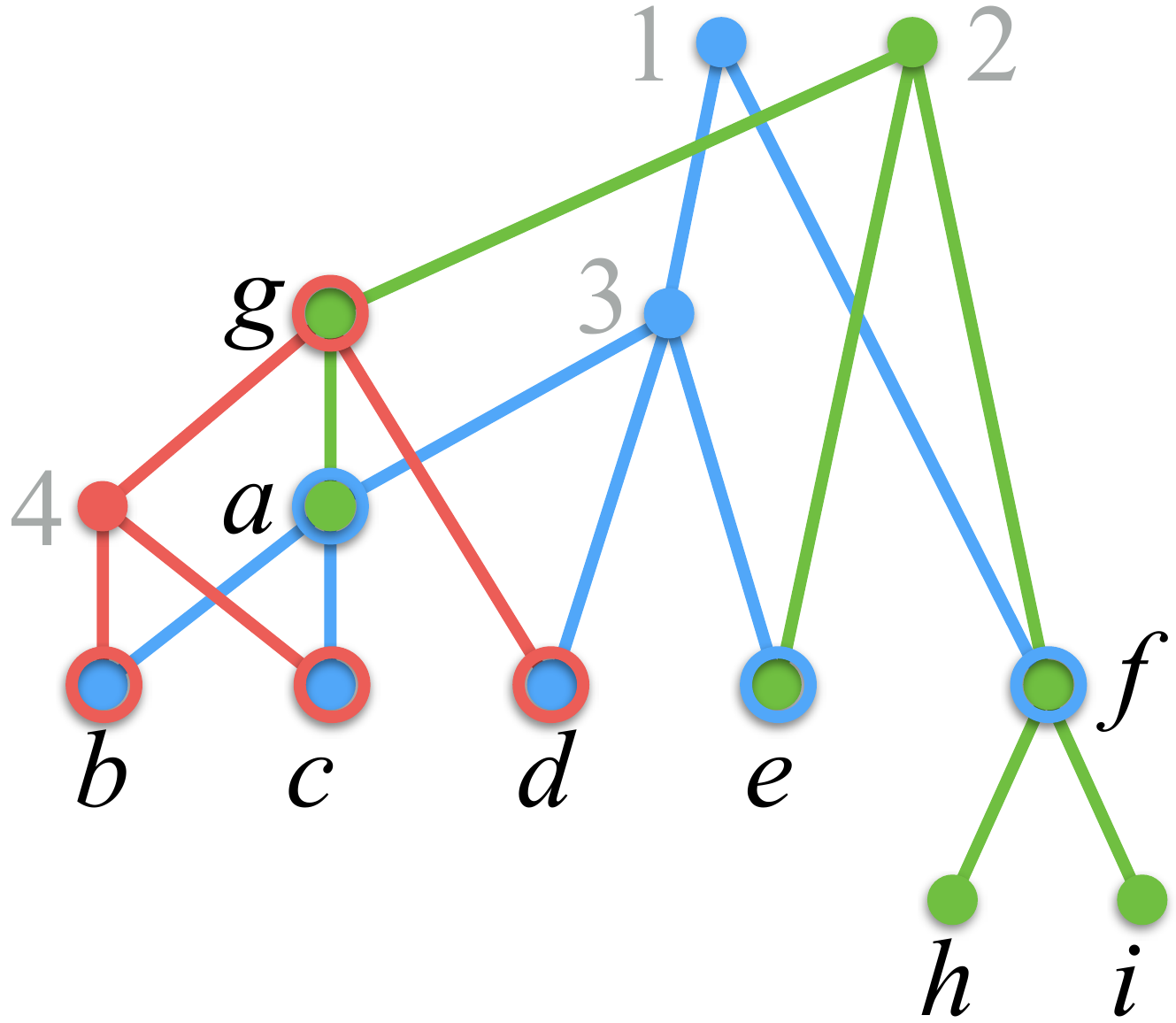}
  \caption{The display graph $\DG$ for the profile of Figure \ref{fig:profile}.} 
  \label{fig:displayGraph}
\end{minipage}
\end{figure}

$\P$ is \emph{ancestrally compatible} if there is a rooted semi-labeled  tree $\T$ that ancestrally displays each of the trees in $\P$.  If $\T$ exists, we say that $\T$ \emph{ancestrally displays} $\P$ (see Figure \ref{fig:supertree}). 

Given a subset $X$ of $L(\P)$, the \emph{restriction} of $\P$ to $X$, denoted $\P|X$, is the profile $\{\T_1|X \cap L(\T_1), \dots, \T_k|X \cap L(\T_k)\}$.  The proof of the following lemma is straightforward.

\begin{lemma}\label{lm:compatSubprofile}
Suppose $\P$ is ancestrally compatible and let $\T$ be a tree that ancestrally displays $\P$.  Then, for any $X \subseteq L(\P)$, $\T | X$ ancestrally displays $\P| X$.
\end{lemma}

A  semi-labeled tree $\T = (T, \phi)$ is \emph{fully labeled} if 
every node in $T$ is labeled.  Suppose $\P$ contains trees that are not fully labeled.  We can convert $\P$ into an equivalent profile $\P'$ of fully-labeled trees as follows.  For each $i \in [k]$, let $l_i$ be the number of unlabeled nodes in $T_i$.  Create a set $L'$ of $n' = \sum_{i \in [k]} l_i$ labels such that $L' \cap L(\P) = \emptyset$.  For each $i \in [k]$ and each $v \in V(T_i)$ such that $\phi_i^{-1}(v) = \emptyset$, make $\phi_i^{-1}(v) = \{\ell\}$, where $\ell$ is a distinct element from $L'$.
We refer to $\P'$ as the \emph{profile obtained by adding distinct new labels to $\P$} (see Figure \ref{fig:profile}).

\begin{lemma}[Daniel and Semple \cite{DanielSemple2004}]\label{lm:fullyL}
Let $\P'$ be the profile obtained by adding distinct new labels to $\P$. Then, $\P$ is ancestrally compatible if and only if $\P'$ is ancestrally compatible. Further, if $\T$ is a semi-labeled phylogenetic tree that ancestrally displays $\P'$, then $\T$ ancestrally displays $\P$.
\end{lemma}

From this point forward, we shall assume that,
for each $i \in [k]$, $\T_i$ is fully and singularly labeled.
By Lemma \ref{lm:fullyL}, no generality is lost in assuming that all trees in $\P$ are fully labeled. 
The assumption that the trees are singularly labeled is inessential; it is only for clarity.  Note that, even with the latter assumption, a tree that ancestrally displays $\P$ is not necessarily singularly labeled.  Figure \ref{fig:supertree} illustrates this fact.

\section{The Display Graph}\label{sec:dispGraph}

The \emph{display graph} of a profile $\P$, denoted $\DG$, is the graph obtained from the disjoint union of the underlying trees $T_1, \dots, T_k$ by identifying nodes that have the same label. Multiple edges between the same pair of nodes are replaced by a single edge. See Figure~\ref{fig:displayGraph}. 

$\DG$ has $O(\MP)$ nodes and edges, and can be constructed in $O(\MP)$ time.  By our assumption 
that all the trees in $\P$ are fully and singularly labeled,
there is a bijection between the labels in $L(\P)$ and the nodes of $\DG$.  Thus, from this point forward, we refer to the nodes of $\DG$ by their labels.
It is easy to see that if $\DG$ is not connected, then $\P$ decomposes into label-disjoint sub-profiles, and that $\P$ is compatible if and only if each sub-profile is compatible.  Thus, we shall assume, without loss of generality, that $\DG$ is connected. 

\vspace{-1.5\parsep}

\paragraph{Positions.} A \emph{position} (for $\P$) is a vector $U = (U(1), \dots, U(k))$, where $U(i) \subseteq L(\T_i)$, for each $i \in [k]$.  Since labels may be shared among trees, we may have $U(i) \cap U(j) \neq \emptyset$, for $i, j \in [k]$ with $i \neq j$.   
For each $i \in [k]$, let $\Desc_i(U) = \{\ell : \ell' \le_{\T_i} \ell, \text{ for some } \ell'  \in U(i)\}$, and let $\Desc_\P(U) = \bigcup_{i \in [k]} \Desc_i(U)$.  

A position $U$ is \emph{valid} if, for each $i \in [k]$, 
\vspace{-0.1cm}
\begin{enumerate}[(V1)]
\vspace{-0.5\parskip}
 \itemsep1pt \parskip0pt \parsep0pt
\item\label{item:v1}
if $|U(i)| \ge 2$, then the elements of $U(i)$ are siblings in $\T_i$ and
\item\label{item:v2}
$\Desc_i(U) = \Desc_\P(U) \cap L(\T_i)$.
\end{enumerate}
\vspace{-1\parsep}

\begin{lemma}\label{lem:V2}
For any valid position $U$, $\P | \Desc_\P(U) = \{\T_1|\Desc_1(U), \dots, \T_k|\Desc_k(U)\}$.
\end{lemma}
\begin{proof}
By (V\ref{item:v2}), we have that $\T_i|\Desc_i(U)$ and $\T_i| \Desc_\P(U) \cap L(\T_i)$ are isomorphic, for each $i \in [k]$. The lemma then follows from the definition of $\P | \Desc_\P(U)$.
\end{proof}

For any valid position $U$, $\DG(U)$ denotes the subgraph of $\DG$ induced by $\Desc_\P(U)$.

\begin{observation}\label{obs:update1}
For any valid position $U$, $\DG(U)$ is the subgraph of $\DG$ obtained by deleting all labels in $V(\DG) \setminus \Desc_\P(U)$, along with all incident edges.
\end{observation}

A valid position of special interest to us is $\Uinit$
, where 
$\Uinit(i) = \phi_i^{-1}(r(T_i))$, for each $i \in [k]$.  That is, $\Uinit(i)$ is a singleton containing only the label of $r(T_i)$.  Thus, in Figure \ref{fig:displayGraph}, $(\Uinit(1), \Uinit(2), \Uinit(3)) = (\{1\}, \{2\}, \{g\})$.
It is straightforward to verify that $\Uinit$ is indeed valid, that $\Desc_\P(\Uinit) = L(\P)$, and that $\DG(\Uinit) = \DG$.  

\vspace{-1.5\parsep}

\paragraph{Semi-universal labels.} Let $U$ be a valid position, and let $\ell$ be a label in $U$.  Then, $\ell$ is \emph{semi-universal in $U$} if  $U(i) = \{\ell\}$, for every $i \in [k]$ such that $\ell \in L(\T_i)$.
It can be verified that in Figure \ref{fig:displayGraph}, labels $1$ and $2$ are semi-universal in $\Uinit$, but $g$ is not, since $g$ is in both $L(\T_2)$ and $L(\T_3)$, but $\Uinit(2) \neq \{g\}$.  

The term ``semi-universal'', borrowed from Pe'er et al.\ \cite{PeerShamirSharan04}, derives from the following fact.
Suppose that $\P$ is ancestrally compatible, that $\T$ is a tree that ancestrally displays $\P$, and that $\ell$ is a semi-universal label for some valid position $U$.  Then, as we shall see, $\ell$ must label the root $u_\ell$ of a subtree of $\T$ that contains all the descendants of $\ell$ in $\T_i$, for every $i$ such that $\ell \in L(\T_i)$.  The qualifier ``semi'' is because this subtree may also contain labels that do not descend from $\ell$ in any input tree, but descend from some other semi-universal label $\ell'$ in $U$ instead. In this case, $\ell'$ also labels $u_\ell$.  This property of semi-universal labels is exploited in both our ancestral compatibility algorithm and its proof of correctness (see Section \ref{sec:testAC}).

For each label $\ell \in L(\P)$, let $k_\ell$ denote the number of input trees that contain label $\ell$.  We can obtain $k_\ell$ for every $\ell \in L(\P)$ in $O(\MP)$ time during the construction of $\DG$.   

\begin{lemma}\label{lem:semiU}
Let $U = (U(1), \dots , U(k))$ be a valid position.  Then, label $\ell$ is semi-universal in $U$ if the cardinality of the set $J_\ell = \{i \in [k] : U(i) = \{\ell\}\}$ equals $k_\ell$.
\end{lemma}
\begin{proof}
By definition, $U(i) = \{\ell\}$, for every $i \in J_\ell$. Since $|J_\ell| = k_\ell$, the lemma follows.
\end{proof}

\vspace{-1.5\parsep}

\paragraph{Successor positions.} 
For every $i \in [k]$ and every $\ell \in L(\T_i)$, let $\Ch_i(\ell)$ denote the set of children of $\ell$ in $L(\T_i)$.  For a subset $A$ of $L(\T_i)$, let $\Ch_i(A) = \bigcup_{\ell \in A} \Ch_i(\ell)$.  
Let $U$ be a valid position, and $S$ be the set of semi-universal labels in $U$.  The \emph{successor of $U$ with respect to $S$} is the position $U'$ defined as follows.  For each $\ell \in S$ and each $i \in [k]$, if $U(i) = \{\ell\}$, then $U'(i) = \Ch_i(\ell)$; otherwise, $U'(i) = U(i)$.  

In Figure \ref{fig:displayGraph}, the set of semi-universal labels in $\Uinit$ is $S = \{1, 2\}$.  Since $\Ch_1(1) = \{3,f\}$ and $\Ch_2(2) = \{e,f,g\}$, the successor of $\Uinit$ is $U'  = (\{3,f\}, \{e,f,g\}, \{g\})$.

\begin{observation}\label{obs:update2}
Let $U$ be a valid position, and let $U'$ be the successor of $U$ with respect to the set $S$ of semi-universal labels in $U$.  Then, $\DG(U')$ can be obtained from $\DG(U)$ by doing the following for each $\ell \in S$: (1) for each $i \in [k]$ such that $U(i) = \{\ell\}$, delete all edges between $\ell$ and $\Ch_i(\ell)$; (2) delete $\ell$.
\end{observation}

Let $U$ be a valid position, and $W$ be a subset of $\Desc_\P(U)$.  Then, $U | W$ denotes the position $(U(1) \cap W, \dots , U(k) \cap W)$.
In Figure \ref{fig:displayGraph}, the components of $\DG(U')$, where $U'$ is the successor of $\Uinit$, are $W_1 = \{3,4,a,b,c,d,e,g\}$ and $W_2 = \{f,h,i\}$.  Thus, $U' | W_1 = (\{3\}, \{e,g\}, \{g\})$ and $U' | W_2 = (\{f\}, \{f\}, \emptyset)$.  We have the following result.

\begin{lemma}\label{lm:childOfU}
Let $U$ be a valid position, and $S$ be the set of all semi-universal labels in $U$.  Let $U'$ be the successor of $U$ with respect to $S$, and let $W_1, W_2, \dots, W_p$ be the label sets of the connected components of $\DG(U')$.  Then, $U' | W_j$ is a valid position, for each $j \in [p]$. 
\end{lemma}

\begin{proof}
It suffices to argue that $U'$ satisfies conditions (V\ref{item:v1}) and  (V\ref{item:v2}).  The lemma then follows from  the fact that the connected components of $\DG(U')$ are label-disjoint.

$U'$ must satisfy condition (V\ref{item:v1}), since $U$ does. Suppose $\ell \in S$.  Then, for each $i \in [k]$ such that $\ell \in L(\T_i)$, $\Desc_i(U') = \Desc_i(U) \setminus \{\ell\}$ and $ \Desc_\P(U') \cap L(\T_i) =(\Desc_\P(U) \cap L(\T_i)) \setminus \{\ell\}$.  Thus, since (V\ref{item:v2}) holds for $U$, it also holds for $U'$.
\end{proof}

\section{Testing Ancestral Compatibility}\label{sec:testAC}

\SLBuild\ (Algorithm \ref{alg:SLBuild}) is our algorithm for testing compatibility of semi-labeled trees.  Its argument, $U$, is a valid position in $\P$ such that $H_\P(U)$ is connected.  Line \ref{alg:findSemi} computes the set $S$ of semi-universal labels in $U$.  If $S$ is empty, then, as argued in Theorem \ref{thm:buildCorrect} below, $\P|\Desc_\P(U)$ is incompatible, and, thus, so is $\P$.  This fact is reported in Line \ref{alg:S0emptyIncompat}.  Line \ref{alg:S0singleton} checks if $S$ contains exactly one label $\ell$, with no proper descendants.  If so, by the connectivity assumption, $\ell$ must be the only element in $\Desc_\P(U)$.  Therefore, Line \ref{alg:S0singletonReturn} simply returns the tree with a single node, labeled $\ell$. Line \ref{alg:updateU} updates $U$, replacing it by its successor with respect to $S$.  Let $W_1, \dots, W_p$ be the connected components of $\DG(U)$ after updating $U$.
By Lemma \ref{lm:childOfU}, $U | W_j$ is a valid position, for each $j \in [p]$. 
Lines \ref{alg:connected}--\ref{alg:recurseEnd} recursively invoke $\SLBuild$ on  $U | W_j$ for each $j \in [p]$, to determine if there is a tree $t_j$ that ancestrally displays $\P| \Desc_\P(U \cap W_j)$.  If any subproblem is incompatible, Line \ref{alg:recurseEnd} reports that $\P$ is incompatible.  Otherwise, Lines \ref{alg:rU}--\ref{alg:finalReturn} assemble the $t_j$s into a single tree that displays $\P|\Desc_\P(U)$, whose root is labeled by the semi-universal labels in the set $S$ of Line \ref{alg:findSemi}.

\begin{algorithm}[t]
\SetAlgoLined
\SetNoFillComment
\DontPrintSemicolon
\KwIn{A valid position $U$ for $\P$ such that $\DG(U)$ is connected.}
\KwOut{A semi-labeled tree that ancestrally displays $\P' = \P|\Desc_\P(U)$, if $\P'$ is ancestrally compatible; \texttt{incompatible} otherwise.
}
Let $S = \{\ell \in U: \ell \text{ is semi-universal in } U\}$ \label{alg:findSemi} \;
\If{$S = \emptyset$\label{alg:S0empty}}
{
	\Return \incompatible\label{alg:S0emptyIncompat}\;
}
\If{$|S| = 1$ and the single element, $\ell$, of $S$ has no proper descendants\label{alg:S0singleton}}
{
	\Return the tree consisting of exactly one node, whose label set is $\{\ell\}$\label{alg:S0singletonReturn}
}
Replace $U$ by the successor of $U$ with respect to $S$.\label{alg:updateU} \; 

Let $W_1, W_2, \dots, W_p$ be the connected components of $\DG(U)$ \label{alg:connected}

\ForEach{$j \in [p]$ \label{alg:recurseBegin}}
{
	Let $t_j = \SLBuild(U | W_j)$\; \label{alg:recurse}
	\If{$t_j$ is not a tree}
	{
		\Return \incompatible \label{alg:recurseEnd}
	}
}
Create a node $r_U$, whose label set is $S$ \label{alg:rU}\;
\Return the tree with root $r_U$ and subtrees $t_1, \dots , t_p$  \label{alg:finalReturn}\;
\caption{\SLBuild$(U)$}\label{alg:SLBuild}
\end{algorithm}

Next, we argue the correctness of $\SLBuild$.

%

\begin{theorem}\label{thm:buildCorrect}
Let $\P = \{\T_1, \dots , \T_k\}$ be a profile and let $\Uinit = (\Uinit(1), \dots, \Uinit(k))$, where, for each $i \in [k]$, $\Uinit(i) = \phi_i^{-1}(r(T_i))$.  Then, $\SLBuild(\Uinit)$
returns either
(i) a semi-labeled tree $\T$ that ancestrally displays $\P$, if $\P$ is ancestrally compatible, or (ii) \texttt{incompatible} otherwise.
\end{theorem}

\begin{proof}
(i) Suppose that $\SLBuild(\Uinit)$ outputs a semi-labeled tree $\T$.  We prove that $\T$ ancestrally displays $\P$. By Lemma \ref{lm:DN}, it suffices to show that $D(\T_i) \subseteq D(\T)$ and  $N(\T_i) \subseteq N(\T)$, for each $i \in [k]$.

Consider any $(\ell,\ell') \in D(\T_i)$.  Then, $\ell$ has a child $\ell''$ in $\T_i$ such that $\ell'' \le_{\T_i} \ell'$.  There must be a recursive call to $\SLBuild(U)$, for some valid position $U$, where $\ell$ is the set $S$ of semi-universal labels obtained in Line \ref{alg:findSemi}.  By Observation \ref{obs:update2}, label $\ell''$, and thus $\ell'$, both lie in one of the connected components of the graph obtained by deleting all labels in $S$, including $\ell$, and their incident edges from $\DG(U)$. It now follows from the construction of $\T$ that $(\ell, \ell') \in D(\T)$.  Thus, $D(\T_i) \subseteq D(\T)$.

Now, consider any $\{\ell,\ell' \} \in N(\T_i)$.  Let $v$ be the lowest common ancestor of $\phi_i(\ell)$ and $\phi_i(\ell')$ in $\T_i$ and let $\ell_v$ be the label of $v$. Then, $\ell_v$ has a pair of children, $\ell_1$ and $\ell_2$ say, in $\T_i$ such that $\ell_1 \le_{\T_i} \ell$, and $\ell_2 \le_{\T_i} \ell'$. Because $\SLBuild(\Uinit)$ returns a tree, there are recursive calls $\SLBuild(U_1)$ and $\SLBuild(U_2)$ for valid positions $U_1$ and $U_2$ such that $\ell_1$ is semi-universal for $U_1$ and $\ell_2$ is semi-universal for $U_2$. We must have $U_1 \neq U_2$; otherwise, $|U_1(i)| = |U_2(i)| \ge 2$, and, thus, neither $\ell_1$ nor $\ell_2$ is semi-universal, a contradiction. Further, it follows from the construction of $\T$ that we must have $\Desc_\P(U_1) \cap \Desc_\P(U_2) = \emptyset$.  Hence, $\ell \parallel_\T \ell'$, and, therefore, $\{\ell,\ell'\} \in N(\T)$.

(ii) Asssume, by way of contradiction, that $\SLBuild(\Uinit)$ returns \texttt{incompatible}, but that $\P$ is ancestrally compatible.  By assumption, there exists a semi-labeled  tree $\T$ that ancestrally displays $\P$.  Since $\SLBuild(\Uinit)$ returns  \texttt{incompatible}, there is a recursive call to $\SLBuild(U)$ for some valid position $U$ such that $U$ has no semi-universal label, and the set $S$ of Line \ref{alg:findSemi} is empty. 

By Lemma~\ref{lm:compatSubprofile}, $\T | \Desc_\P(U)$ ancestrally displays $\P|\Desc_\P(U)$.  Thus, by Lemma \ref{lem:V2}, $\T | \Desc_\P(U)$ ancestrally displays $\T_i | \Desc_i(U)$, for every $i \in [k]$.  Let
$\ell$ be any label in the label set of the root of $\T | \Desc_\P(U)$.  Then, for each $i \in [k]$ such that $\ell \in L(\T_{i})$, $\ell$ must be the label of the root of $\T_i|\Desc_i(U)$.
Thus, for each such $i$, $U(i) = \{\ell\}$.  Hence, $\ell$ is semi-universal in $U$, a contradiction.
\end{proof}

\section{Implementation}\label{sec:implementation}

Here we describe an efficient implementation of $\SLBuild$.  We focus on two key aspects: finding semi-universal labels in Line \ref{alg:findSemi}, and updating $U$ and $\DG(U)$ in Lines \ref{alg:updateU} and \ref{alg:connected}.  

By Observation \ref{obs:update1}, at each recursive call, $\SLBuild$ deals with a graph obtained from $\DG$ through edge and node deletions.  To handle these deletions efficiently, we represent $\DG$ using the dynamic graph connectivity data structure of  Holm et al.\ \cite{HolmLichtenbergThorup:2001}, which we refer to as \emph{HDT}.  HDT allows us to maintain the list of nodes in each component, as well as the number of these nodes so that, if we start with no edges in a graph with $N$ nodes, the amortized cost of each update is $O(\log^2 N)$.  Since $\DG$ has $O(\MP)$ nodes, each update takes $O(\log^2 \MP)$ time.  
The total number of edge and node deletions performed by $\SLBuild(\Uinit)$ --- including all deletions in the recursive calls --- is at most the total number of edges and nodes in $\DG$, which is $O(\MP)$.  HDT allows us to maintain connectivity information throughout the entire algorithm in $O(\MP \log^2 \MP)$ time.

As deletions are performed on $\DG$, $\SLBuild$ maintains three data fields for each connected component $Y$ that is created: $Y.\WEIGHT$, $Y.\MAP$, and $Y.
\SEMI$.  It also maintains a field $\ell.\cnt$, for each $\ell \in L(\P)$.  
\vspace{-0.1cm}
\begin{enumerate}
\vspace{-0.5\parskip}
 \itemsep1pt \parskip0pt \parsep0pt
\item
$Y.\WEIGHT$ equals $\sum_{\ell \in Y} k_\ell$.
\item
$Y.\MAP$ is a map from a set $J_Y \subseteq [k]$ to a set of nonempty subsets of $Y \cap L(T_i)$.  For each $i \in J_Y$, $Y.\MAP(i)$ denotes the set associated with $i$.
\item
$\ell.\cnt$ equals the cardinality of the set
$\{i \in [k] : Y.\MAP(i) \text{ is defined and } Y.\MAP(i) = \{\ell\}\}$. (Recall that $k_\ell$ is the number of input trees that contain $\ell$.)
\item
$Y.\SEMI$ is a set containing all labels $\ell \in Y$ such that $\ell.\cnt = k_\ell$.   
\end{enumerate}
\vspace{-0.1cm}
Informally, each set $Y.\MAP(i)$ corresponds to a non-empty $U(i)$; $Y.\SEMI$ corresponds to the semi-universal labels in $Y$.  Next, we formalize these ideas.

At the start of the execution of $\SLBuild(U)$ for any valid position $U$, $\DG(U)$ has a single connected component, $Y_U = \Desc_\P(U)$.  Our implementation maintains the following invariant. 
\begin{description}
\item[INV:]
At the beginning of the execution of $\SLBuild(U)$, $Y_U.\MAP(i) = U(i)$ for each $i \in [k]$ such that $U(i) \neq \emptyset$, and $Y_U.\MAP(i)$ is undefined for each $i \in [k]$ such that $U(i) = \emptyset$.
\end{description}
Thus, $\ell.\cnt$ equals the number of indices $i \in [k]$ such that $U(i) = \{\ell\}$.  Along with Lemma \ref{lem:semiU}, INV implies that, at the beginning of the execution of $\SLBuild(U)$, $Y_U.\SEMI$ contains precisely the semi-universal labels of $U$.
Thus, the set $S$ of line \ref{alg:findSemi} of $\SLBuild(U)$ can be retrieved in $O(1)$ time.

To establish INV for the initial valid position $\Uinit$, we proceed as follows.  By assumption, $\DG(\Uinit)$ has a single connected component, $\Yinit =  L(\P)$.  
Since $\DG(\Uinit)$ equals $\DG$, we initialize data fields 1--4 for $\Yinit$ during the construction of $\DG$. 
$\Yinit.\WEIGHT$ is simply $\sum_{\ell \in L(\P)} k_\ell$. For each $i \in [k]$, $\Yinit.\MAP(i)$ is $\{\ell\}$, where $\ell$ is the label of the root of $T_i$.  We initialize the \cnt\ fields as follows.  
First, set $\ell.\cnt$ to $0$ for all $\ell \in L(\P)$.  Then,  iterate through each $i \in [k]$, incrementing $\ell.\cnt$ by one if $\Yinit.\MAP(i) = \{\ell\}$.  Finally, $\Yinit.\SEMI$ consists of all $\ell \in \Uinit$ such that $\ell.\cnt = k_\ell$.  All data fields can be initialized in $O(\MP)$ time.

We now focus on Lines \ref{alg:updateU} and \ref{alg:connected} of $\SLBuild$. By Observation \ref{obs:update2}, we can update $U$ and $\DG(U)$ jointly as follows. We use a temporary variable $\GBNT$.  Prior to executing Line \ref{alg:updateU}, we set $\GBNT = \DG(U)$.
Then, we successively consider each label $\ell \in S$, and perform two steps: (i) initialize data fields 1--4 in preparation for the deletion of $\ell$ and (ii) delete from $\GBNT$ the edges incident on $\ell$ and then $\ell$ itself, updating data fields 1--4 as necessary, to maintain INV.  After these steps are executed, $\GBNT$ will equal $\DG(U)$ for the new set $U$ created by Line \ref{alg:updateU}.  Steps (i) and (ii) are done by $\InitLists(\ell)$ (Algorithm \ref{alg:initMarks}) and $\DelLabel(\ell)$ (Algorithm \ref{alg:delLabel}), respectively. 

Lines \ref{alg:delSemi}--\ref{alg:addPairMap} of $\InitLists(\ell)$ initialize $Y.\MAP$ and $Y.\SEMI$ to reflect the fact that label $\ell \in S$ is leaving $U(i)$, for each $i \in [k]$ such that $i \in L(\T_i)$, to be replaced by its children in $\T_i$, and will no longer be semi-universal.  Lines \ref{alg:initSingleton}--\ref{alg:forAddAlphaSemi} are needed to update certain $\cnt$ fields due to the possibility that singleton sets $Y.\MAP(i)$ may be created in the preceding steps.  The number of operations on $Y.\MAP$ performed by $\InitLists(\ell)$ is $O(\sum_{i \in [k] : \ell \in L(\T_i)} |\Ch_i(\ell)|)$; i.e., it is proportional to the total number of children of $\ell$ in all the input trees.  Since $\ell$ is considered only once, the total number of operations on $\MAP$ fields of the various sets $Y$ considered over the entire execution of $\SLBuild(\Uinit)$ is $O(\MP)$.  The number of updates of $Y.\MAP$ done by $\InitLists(\ell)$ is at most $k_\ell$; the total over all $\ell$ considered by $\SLBuild(\Uinit)$ over its entire execution is $O(\MP)$.

\begin{algorithm}
\SetAlgoLined\SetNoFillComment

\DontPrintSemicolon
	Delete $\ell$ from $Y.\SEMI$ \label{alg:delSemi}\;
	\ForEach {$i \in [k]$ such that $\ell \in L(\T_i)$\label{alg:forEachIContainingL}}
	{
		Delete $Y.\MAP(i)$ \;
		\ForEach {$\alpha \in \Ch_i(\ell) $\label{alg:forEachAlpha}}
		{
			Add $\alpha$ to $Y.\MAP(i)$ \label{alg:addPairMap} \; 
		}
		\If{$Y.\MAP(i)$ is a singleton\label{alg:initSingleton}}
		{
			Let $\beta$ be the single element in $Y.\MAP(i)$ \;
			Set $\beta.\cnt = \beta.\cnt +1$ \;
			\lIf{$\beta.\cnt = k_{\beta}$}
			{
				add $\beta$ to $Y.\SEMI$\label{alg:forAddAlphaSemi}
			}
		}
	}

\caption{\InitLists$(\ell)$}\label{alg:initMarks}
\end{algorithm}

$\DelLabel(\ell)$ begins 
by consulting HDT to identify the connected component $Y$ that currently contains $\ell$.  The loop in Lines \ref{alg:forScanEdges}--\ref{alg:endForScanEdges} successively deletes each edge between $\ell$ and a child $\alpha$ of $\ell$, updating the appropriate data fields for the resulting connected components.  Line \ref{alg:testConn} queries the HDT data structure to determine whether deleting $(\ell,\alpha)$ splits $Y$ into two components. If $Y$ remains connected, no updates are needed --- the \textbf{continue} statement skips the rest of the current iteration and proceeds directly to the next.  Otherwise, $Y$ is split into two parts $Y_1$ and $Y_2$.  $\DelLabel$ uses a weighted version of the technique of scanning the smaller component \cite{EvenShiloach:1981}.  Line \ref{alg:smallerWeight} identifies which of the two components has the smaller \WEIGHT\ field; without loss of generality, it assumes that $Y_1.\WEIGHT \le Y_2.\WEIGHT$.  Lines \ref{alg:delInitY1} and \ref{alg:delInitY2} initialize $Y_1.\MAP$ and $Y_1.\SEMI$ to \NULL\ and $Y_2.\MAP$ and $Y_2.\SEMI$ to the corresponding fields of $Y$.   Lines \ref{alg:forBetaY1}--\ref{alg:addToJ}, scan each label $\beta$ in $Y_1$, from $Y_2.\MAP(i)$ to $Y_1.\MAP(i)$, for every $i$ such that $\beta \in L(\T_i)$.
Set $J$, updated in Line \ref{alg:addToJ}, keeps track of the indices $i$ such that $Y_1.\MAP(i)$ and $Y_2.\MAP(i)$ are modified.  Lines \ref{alg:forUpdateSemi}--\ref{alg:endForScanEdges} iterate through $J$ to determine if any new singleton sets were created in either $Y_1$ or $Y_2$.  
This operation requires at most one update in each of $Y_1.\SEMI$ and $Y_2.\SEMI$; each update takes $O(1)$ time.   After all edges incident on $\ell$ are deleted, $\ell$ itself is deleted (Line \ref{alg:delEll}).

The preceding description of $\DelLabel(\ell)$ omits the updating of the $\WEIGHT$ fields of the connected components created by an edge deletion. This is done before Line \ref{alg:smallerWeight}, by (again) using the technique of scanning the smaller component.  We consult HDT to determine which of $Y_1$ and $Y_2$ has fewer labels.  Assuming, without loss of generality, that $|Y_1| < |Y_2|$, compute $Y_1.\WEIGHT$ in a sequential scan of $Y_1$.  Then, $Y_2.\WEIGHT = Y.\WEIGHT - Y_1.\WEIGHT$.

\begin{algorithm}[t]
\SetAlgoLined\SetNoFillComment

\DontPrintSemicolon
	Let $Y$ be the connected component of $\GBNT$ that contains $\ell$ \label{alg:findY}\;
	\ForEach {$\alpha \in \Ch(\ell)$\label{alg:forScanEdges}}
	{
		Delete edge $\{\ell, \alpha\}$ from $\GBNT$ \label{alg:deleteEdge}\;
		\lIf{$Y$ remains connected\label{alg:testConn}}{\continue}
		Let $Y_1, Y_2$ be the connected components of $\GBNT$; assume that $Y_1.\WEIGHT \le Y_2.\WEIGHT$ \label{alg:smallerWeight} \;
		Let $Y_1.\MAP = \NULL$ and $Y_1.\SEMI = \NULL$ \label{alg:delInitY1}\;
		Let $Y_2.\MAP = Y.\MAP$ and $Y_2.\SEMI = Y.\SEMI$ \label{alg:delInitY2}\;
		Let $J = \emptyset$ \;
		\ForEach {$\beta \in Y_1$\label{alg:forBetaY1}}
		{
			\ForEach{$i \in [k]$ such that $\beta \in L(\T_i)$}
			{
					Move $\beta$ from $Y_2.\MAP(i)$ to $Y_1.\MAP(i)$ \label{alg:MoveBeta} \;
					$J = J \cup \{i\}$ \label{alg:addToJ}
			}
		}
		\ForEach {$i \in J$\label{alg:forUpdateSemi}}
		{
			\ForEach{$j \in \{1,2\}$}
			{
				\If{$Y_j.\MAP(i) = \emptyset$}
				{
					Delete $Y_j.\MAP(i)$
				}
				\ElseIf{$Y_j.\MAP(i)$ is a singleton}
				{
					Let $\gamma$ be the single element in $Y_j.\MAP(i)$ \;
					$\gamma.\cnt = \gamma.\cnt + 1$ \;
					\lIf{$\gamma.\cnt = k_\gamma$}
					{
						add $\gamma$ to $Y_j.\SEMI$ \label{alg:endForScanEdges}
					}
				}
			}
		}
	}
	Delete $\ell$ from $\GBNT$ \label{alg:delEll}
\caption{$\DelLabel(\ell)$}\label{alg:delLabel}
\end{algorithm}

Let us track the number of operations on $\MAP$ fields in Lines \ref{alg:forBetaY1}--\ref{alg:addToJ} of $\DelLabel(\ell)$ that can be attributed to some specific label $\beta \in L(\P)$ over the entire execution of $\SLBuild(\Uinit)$.  
Each execution of Lines \ref{alg:forBetaY1}--\ref{alg:addToJ} for $\beta$ performs $k_\beta$ operations on $\MAP$ fields.
Let $w_r(\beta)$ be the weight of the connected component containing $\beta$ at the beginning of the loop of Lines \ref{alg:forBetaY1}--\ref{alg:addToJ}, at the $r$th time that $\beta$ is considered in those lines; thus, $w_0(\beta) \le \sum_{\ell \in L(\P)} k_\ell$.  Then, $w_r(\beta) \le w_0(\beta)/2^r$.  The reason is that we only consider $\beta$ if (i) $\beta$ is contained in one of the two components that result from
deleting an edge in Line \ref{alg:deleteEdge} and (ii) the component containing $\beta$ has the smaller weight of the two components.  Thus, the number of times $\beta$ is considered in Lines \ref{alg:forBetaY1}--\ref{alg:addToJ}  over the entire execution of $\SLBuild(\Uinit)$ is $O(\log  w_0(\beta))$, which is $O(\log \MP)$, since $w_0(\beta) = O(\MP)$. Therefore, the total number of updates of $\MAP$ fields over all labels is $O(\log \MP \cdot \sum_{\ell \in L(\P)} k_\ell)$, which is $O(\MP \log \MP)$. 
It can be verified that the number of updates to \cnt\ and \SEMI\ fields is also $O(\MP \log \MP)$. 
A similar analysis shows that the total time to update $\WEIGHT$ fields over all edge deletions performed by  $\SLBuild(\Uinit)$ is $O(\MP \log \MP)$. 

To summarize, the work done by $\SLBuild(\Uinit)$ consists of three parts: (i) initialization, (ii) maintaining connected components, and (iii) maintaining the \WEIGHT, \MAP, and \SEMI, and fields for each connected component, as well as $\ell.\cnt$ for each label $\ell$. Part (i) takes $O(\MP)$ time. Part (ii) involves $O(\MP)$ edge and node deletions on the HDT data structure, at an amortized cost of $O(\log^2 \MP)$ per deletion. Part (iii) requires a total of $O(\MP \log \MP)$ updates to the various fields.  Using data structures that take logarithmic time per update, leads to our main result.  

\begin{theorem}\label{thm:SLBuildA}
\SLBuild\ can be implemented so that \SLBuild$(\Uinit)$ runs in $O(\MP \log^2 \MP)$ time.
\end{theorem}

\section{Discussion\label{sec:discussion}}

Like our earlier algorithm for compatibility of ordinary phylogenetic trees, the more general algorithm presented here, \SLBuild, is a polylogarithmic factor away from optimality (a trivial lower bound is $\Omega(\MP)$, the time to read the input).  \SLBuild\ has a linear-space implementation, using the results of Thorup \cite{Thorup2000}.  A question to be investigated next is the performance of the algorithm on real data.  Another important issue is integrating our algorithm into a synthesis method that deals with incompatible profiles.


\end{document}